\documentclass[12pt]{article}

\usepackage{times}
\usepackage{fullpage}
\usepackage{amsfonts,amssymb,amsmath,xspace}
\usepackage{latexsym}
\usepackage{url}
\usepackage{verbatim}
\usepackage[pdftex]{color}
\usepackage[breaklinks]{hyperref}

\newcommand{\Q}{\mathbb{Q}}

\newcommand{\F}{\mathbb{F}}

\newcommand{\floor}[1]{\lfloor #1 \rfloor}
\newcommand{\ceil}[1]{\left\lceil #1 \right\rceil}

\newcommand{\braket}[2]{\langle #1|#2 \rangle}
\newcommand{\ketbra}[2]{| #1\rangle \langle #2|}

\def\01{\{0,1\}}

\newcommand{\Tr}{\mathrm{Tr}}

\newcommand{\IP}{\mathrm{IP}}
\newcommand{\diag}{\mathrm{diag}}
\newcommand{\conv}{\mathrm{conv}}

\newcommand{\ignore}[1]{}

\newtheorem{theorem}{Theorem}
\newtheorem{lemma}[theorem]{Lemma}
\newtheorem{corollary}[theorem]{Corollary}
\newtheorem{proposition}[theorem]{Proposition}
\newtheorem{definition}[theorem]{Definition}

\newtheorem{fact}[theorem]{Fact}

\newcommand{\thmref}[1]{\hyperref[#1]{{Theorem~\ref*{#1}}}}
\newcommand{\lemref}[1]{\hyperref[#1]{{Lemma~\ref*{#1}}}}
\newcommand{\corref}[1]{\hyperref[#1]{{Corollary~\ref*{#1}}}}
\newcommand{\eqnref}[1]{\hyperref[#1]{{Equation~(\ref*{#1})}}}
\newcommand{\claimref}[1]{\hyperref[#1]{{Claim~\ref*{#1}}}}
\newcommand{\remarkref}[1]{\hyperref[#1]{{Remark~\ref*{#1}}}}
\newcommand{\propref}[1]{\hyperref[#1]{{Proposition~\ref*{#1}}}}
\newcommand{\factref}[1]{\hyperref[#1]{{Fact~\ref*{#1}}}}
\newcommand{\defref}[1]{\hyperref[#1]{{Definition~\ref*{#1}}}}
\newcommand{\exampleref}[1]{\hyperref[#1]{{Example~\ref*{#1}}}}
\newcommand{\hypref}[1]{\hyperref[#1]{{Hypothesis~\ref*{#1}}}}
\newcommand{\secref}[1]{\hyperref[#1]{{Section~\ref*{#1}}}}
\newcommand{\chapref}[1]{\hyperref[#1]{{Chapter~\ref*{#1}}}}
\newcommand{\apref}[1]{\hyperref[#1]{{Appendix~\ref*{#1}}}}
\newcommand\rank{\mbox{\tt {rank}}\xspace}
\newcommand\prank{\mbox{\tt {rank}$_{\tt psd}$}\xspace}

\newcommand\rootrank{\mbox{\tt {rank}$_{\sqrt{ }}$}\xspace}

\newenvironment{proof}[1][Proof: ]
{\noindent {\bf #1}}
{{\hfill $\Box$}\\
 \smallskip}

\begin{document}

\title{The square root rank of the correlation polytope is exponential}
\author{Troy Lee\thanks{School of Physics and Mathematical Sciences, Nanyang Technological
University and Centre for Quantum Technologies, Singapore.
Email: \{troyjlee, weizhaohui\}@gmail.com}
\and Zhaohui Wei\footnotemark[1]}
\date{}
\maketitle

\begin{abstract}
The square root rank of a nonnegative matrix $A$ is the minimum rank of a matrix $B$ such that
$A=B \circ B$, where $\circ$ denotes entrywise product.  We show that the square root rank of the
slack matrix of the correlation polytope is exponential.  Our main technique is a way to lower bound
the rank of certain matrices under arbitrary sign changes of the entries using properties of the
roots of polynomials in number fields.  The square root rank is an upper bound on the positive
semidefinite rank of a matrix, and corresponds the special case where all matrices in the
factorization are rank-one.
\end{abstract}

\section{Introduction}
The square root rank of a nonnegative matrix $A$ is the minimum rank of a matrix $B$ such that
$A=B \circ B$, where $\circ$ denotes the entrywise product.  The freedom of the matrix $B$ is to multiply each
entry $\sqrt{A(i,j)}$ by $\pm 1$ in an effort to decrease the rank, and this substantial freedom is what makes
showing lower bounds on the square root rank difficult.  It is known that the problem of verifying if the square root
rank is less than a given value is NP-hard \cite{FawziGouveiaParriloRobinsonThomas14}.

The motivation for studying square root rank is that it is an upper
bound on the positive semidefinite rank
\cite{GouveiaParriloThomas12, Zhang12}.  A positive semidefinite
(PSD) factorization of an $m$-by-$n$ nonnegative matrix $A$ of size
$r$ is given by $r$-by-$r$ real PSD matrices $E_1, \ldots,
E_m,F_1,\ldots, F_n$ such that $A(i,j)=\Tr(E_i F_j)$.  The square
root rank exactly corresponds to the minimum size of a PSD
factorization where all the PSD matrices are rank-one.

The positive semidefinite
rank has been defined relatively recently in the context of combinatorial optimization.
Many combinatorial
optimization problems can be represented as optimizing a linear function over a polytope $P$ formed by the
convex hull of feasible solutions.  A natural way to approach this problem is via linear programming and here
the number of constraints in the linear program is given by the number of facets of $P$.

A remarkable fact is that sometimes there is a higher dimensional polytope $Q$ with fewer facets that projects to $P$.
In this way, the original optimization problem can be transferred to an easier optimization problem over $Q$.
The polytope $Q$ is called a linear extension of $P$ and the minimum number of facets of such a $Q$ is the linear
extension complexity of $P$.

A classic paper of Yannakakis beautifully characterizes the linear
extension complexity \cite{Yannakakis91}. For a polytope $P$ with
facet inequalities $a_i x \le b_i$ and vertex set $V=\{v_j\}$, the
slack matrix of $P$ is the matrix with the $(i,j)$ entry equal to
$b_i - a_i v_j$.  Yannakakis showed that the linear extension
complexity of $P$ is equal to the nonnegative rank of the slack
matrix of $P$. The nonnegative rank of a nonnegative matrix $A$ is
the minimum number of nonnegative rank-one matrices that sum to $A$.
Answering a long standing open question, Fiorini et al.
\cite{FioriniMassarPokuttaTiwaryDewolf2012} showed exponential lower
bounds on the linear extension complexity of many polytopes of
interest, including the correlation and Traveling Salesman
polytopes.  Rothvo{\ss} followed this by showing an exponential
lower bound on the linear extension complexity of the matching
polytope \cite{Rothvoss14}, even though finding a maximum matching
can be done efficiently.

As semidefinite programming is more powerful than linear programming
it is natural to ask the same questions for semidefinite extensions.
A semidefinite extension of a polytope $P$ is an affine slice of the
cone of $n$-by-$n$ positive semidefinite matrices that projects to
$P$.  The proof of Yannakakis can be adapted to this setting, and
Gouveia et al. showed that the semidefinite extension complexity of
$P$ is equal to the PSD rank of the slack matrix of $P$
\cite{GouveiaParriloThomas12}.

The correlation polytope $\text{COR}_n$ is the convex hull of the
rank-one boolean matrices $xx^T$ for $x \in \{0,1\}^n$.  The
correlation polytope is closely related to the cut polytope and has
proven to be the most convenient polytope to study for extension
complexity lower bounds.  In a very recent breakthrough, Lee et al.
have given exponential lower bounds on the PSD-rank of the slack
matrix of the correlation polytope \cite{LeeSR14}.  Before this, no
nontrivial bounds were known on the PSD-rank of the correlation
polytope, and indeed no techniques had been developed to approach
this problem.

Our main result is a lower bound of $3^{n/3-1}$ on the square root rank of the slack matrix of $\text{COR}_n$.  We
do this by showing a severe algebraic limitation to factorizations of the form $A=B \circ B$.  Our techniques are
fairly general and apply to many other matrices, even those that actually have small PSD rank.  Though the main
open
problem of showing an exponential lower bound on the PSD-rank of the correlation polytope has now been
answered, our techniques may still be interesting as many constructions of PSD factorizations are actually
rank-one and so their size corresponds to square root rank.

\section{Preliminaries}

\subsection{Notations and definitions}
Let $[n]=\{1, 2, \ldots, n\}$.
As usual, we refer to the fields of rational, real, and complex numbers as
$\mathbb{Q}, \mathbb{R}$, and $\mathbb{C}$.  A subfield of the real numbers is a field $\F \subseteq R$
that is a subset of the real numbers.  Any subfield of the real numbers contains the rationals $\Q$.

The correlation polytope $\text{COR}_n$ is the convex hull of matrices of the form $xx^T$, where
$x\in\{0,1\}^n$ is a column vector, and $x^T$ is the transpose of
$x$. In other words,
$\text{COR}_n=\conv\{xx^T \in \mathbb{R}^{n\times n} : x\in\{0,1\}^n\}$.

For an $m$-by-$n$ matrix $A$, we refer to the $(i,j)$ entry as
$A(i,j)$.  We use $\circ$ for the entrywise product, that is $(A
\circ B)(i,j)=A(i,j)B(i,j)$. We denote the rank of $A$ by
$\rank(A)$, and if $m=n$, denote the trace as $\Tr(A)=\sum_i
A(i,i)$. If all the entries of $A$ are either zero or positive, we
call $A$ a \emph{nonnegative matrix}.

If a matrix $A$ is nonnegative, its \emph{nonnegative rank}, denoted $\rank_+(A)$, is the minimum number of
rank-one nonnegative matrices that sum to $A$.
The positive semidefinite rank is defined as follows.
\begin{definition}
Let $A$ be a nonnegative $m$-by-$n$ matrix.  A  \emph{positive
semidefinite factorization} (over $\mathbb{R}$) of $A$ of size $r$
is given by $r$-by-$r$ real positive semidefinite matrices $E_1,
\ldots, E_m \in \mathbb{R}^{r\times r}$ and $F_1,\ldots, F_n \in \mathbb{R}^{r \times r}$
satisfying $A(i,j)=\Tr(E_i F_j)$ for all $i=1, \ldots, m$ and $j=1, \ldots, n$.  The \emph{positive
semidefinite rank} denoted $\prank(A)$ of $A$
is the smallest integer~$r$ such that $A$ has a PSD-factorization of
size~$r$.
\end{definition}

The main quantity of interest in this paper is the square root rank.
\begin{definition}
Let $A$ be a nonnegative $m$-by-$n$ matrix. The square root rank of
$A$ is the minimum rank of an $m$-by-$n$ matrix $B$ with $A=B \circ
B$, and is denoted $\rootrank(A)$.
\end{definition}
For a nonnegative matrix $A$, we will use $\sqrt{A}$ for the entrywise square root of $A$, that is
$\sqrt{A}(i,j)=\sqrt{A(i,j)}$.

\subsection{Basic facts about PSD-rank}
In this section we discuss some basic results about the PSD-rank.  Nearly all of these results can be found in the
excellent survey \cite{FawziGouveiaParriloRobinsonThomas14}.
The first fact is an easy lower bound on PSD-rank in terms of the normal rank.
\begin{fact}\label{ft:trivial}
Let $A$ be a nonnegative matrix.  Then
$\prank(A)\geq\sqrt{\rank(A)}$.
\end{fact}

It is also easy to see that the nonnegative rank is an upper bound on the PSD-rank.
\begin{fact}\label{ft:rk+}
Let $A$ be nonnegative matrix.  Then $\prank(A) \le \rank_+(A)$.
\end{fact}
A nonnegative rank factorization corresponds to a PSD-factorization by diagonal matricies.

At the other end of the spectrum, one can consider
PSD-factorizations by rank-one matrices. An equivalent
characterization of the square root rank is the minimal size of a
PSD-factorization by rank-one PSD matrices.
\begin{fact}(\cite{GouveiaRobinsonThomas13})\label{ft:rankone}
Let $A$ be a nonnegative $m$-by-$n$ matrix. Then
$\rootrank(A)$ is equal to the minimum size of a PSD factorization $A(i,j) = \Tr(E_i F_j)$ where all
the PSD matrices $E_1, \ldots, E_m, F_1, \ldots, F_n$ are rank-one.
\end{fact}
%

In particular, this characterization shows the following.
\begin{corollary}[\cite{GouveiaParriloThomas12, Zhang12}]
\label{ft:upperbound}
For a nonnegative matrix $A$
\[
\prank(A) \le \rootrank(A)
\]
\end{corollary}
It can sometimes be difficult to see how to use the power of positive semidefinite factorizations to show upper
bounds on the PSD-rank.  For this reason, many upper bounds on PSD-rank simply use
\factref{ft:upperbound}.  This upper bound can also be tight in a nontrivial way.

An example of this can be seen with the inner product matrix. This matrix has been extensively
studied in communication complexity and is defined as
$\IP_n(x,y)=\sum^n_{i=1}x_iy_i \bmod 2$ for $x,y \in \{0,1\}^n$.  Letting $N=2^n$ be the size of the matrix
$IP_n$,  Lee et al. \cite{LeeWeideWolf14} prove that
$\rootrank(IP_n)\leq 2\sqrt{N}$. Note that $IP_n$ is full-rank, thus
according to \factref{ft:trivial} it holds that $\prank(IP_n)\geq
\sqrt{N}$, which implies that the upper bound given by
\factref{ft:upperbound} is tight in this case up to a small constant
factor.  Note that $\sqrt{IP_n} = IP_n$ and thus has high rank---the construction crucially uses the freedom of
toggling the sign of each entry.

The usual rank of a matrix $A$ is equal to the minimal number of rank-one matrices that sum to $A$.  There is no
known analogous ``decomposition'' formulation for the PSD-rank.  The following lemma, however, does give
an approximate
characterization of PSD-rank in terms of a decomposition of matrices with rank-one PSD factorizations.
We first learned of this lemma from Ronald de Wolf \cite{deWolf12}.
\begin{lemma}\label{lem:decomposition}
Suppose that the PSD-rank of $A$ is $d$.  Then there is a decomposition
\[
A = \sum_{i=1}^{d^2} N_i \circ N_i
\]
where each $N_i$ is of rank at most $d$.
\end{lemma}
\begin{proof}
Suppose $E_x$ and $F_y$ is an optimum PSD-factorization for $A$,
i.e., $A(x,y)=\Tr(E_xF_y)$ and $E_x,F_y$ are $d$-by-$d$ PSD matrices. For each $x$ and $y$, let
$E_x=\sum_{k_1=1}^d \ketbra{\alpha_x^{k_1}}{\alpha_x^{k_1}}$ and
$F_y=\sum_{k_2=1}^d\ketbra{\beta_y^{k_2}}{\beta_y^{k_2}}$ be
spectral decompositions of $E_x$ and $F_y$. Then
\[
A(x,y) =
\sum_{k_1,k_2=1}^{d}|\braket{\alpha_x^{k_1}}{\beta_y^{k_2}}|^2.
\]
For each $k_1$ and $k_2$, define a matrix $A_{k_1,k_2}$ by setting
the entries as
$A_{k_1,k_2}(x,y)=\braket{\alpha_x^{k_1}}{\beta_y^{k_2}}$.  Then its
rank is at most $d$ and $A=\sum_{k_1,k_2=1}^{d}
A_{k_1,k_2} \circ A_{k_1,k_2}$.
\end{proof}

\section{Square root rank of the correlation polytope}
In this section we prove that the square root rank of the correlation polytope $\text{COR}_n$ is at least
$3^{n/3-1}$.
Our approach uses an algebraic method to lower bound the rank of certain matrices based on the roots of
their characteristic polynomials.

For a univariate polynomial $q(x)$ with real coefficients, a familiar theorem states that the multiplicity of $a+bi$
and $a-bi$ as
roots of $q$ is the same.  The key to our lower bounds will be the following generalization of this to subfields of the
real numbers.  A similar statement can be found in any textbook on Galois theory, see for example Lemma 5.6 of
\cite{Stewart04}.
We include a proof for completeness.
\begin{theorem}
\label{thm:key}
Let $\F$ be a subfield of the real numbers and $p$ a prime such that $\sqrt{p} \not \in \F$.  Then for any univariate
polynomial $q(x)$ with coefficients in $\F$ the multiplicity of $\sqrt{p}$ and $-\sqrt{p}$ as roots of $q$ is the same.
\end{theorem}
\begin{proof}
Let $m(x)=x^2 - p$.  We first see that $m$ divides any polynomial that has a root at $\sqrt{p}$.
Let $h$ be a polynomial with coefficients in $\F$ and a root at $\sqrt{p}$.  By the division algorithm
$h= m g +r$ where $r$ is a polynomial with coefficients in $\F$ that is either the constant zero function or
a polynomial of degree at most $1$.  If $h(\sqrt{p})=0$ then $r$ must be the zero function, since no
nonzero polynomial of degree at most $1$ polynomial can have a root at $\sqrt{p}$, as $\sqrt{p} \not \in \F$.
The same argument holds for a polynomial with a root at $-\sqrt{p}$.

Now let $k$ be the largest power of $m$ that divides $q$, that is such that
$q=m^k h$ for some polynomial $h$ with coefficients in $\F$ and $m^{k+1}$ does not divide $q$ in $\F$.  By
definition, $m$ does not divide $h$, thus
$h$ cannot have a root at $\sqrt{p}$ or $-\sqrt{p}$.  This shows that the multiplicity of both $\sqrt{p},-\sqrt{p}$ as roots
of $q$ is $k$, and is
the same.
\end{proof}

We can use \thmref{thm:key} to show a lower bound on the rank of certain matrices in the following way.
\begin{theorem}
\label{thm:rank}
Let $\F$ be a subfield of the real numbers and $p$ a prime such that $\sqrt{p} \not \in \F$.  Let $A \in \F^{N \times N}$.
Then $\rank(\sqrt{p}I+A) \ge \ceil{\tfrac{N}{2}}$.
\end{theorem}

\begin{proof}
We will show that the nullity of $\sqrt{p}I + A$ is at most $\floor{\tfrac{N}{2}}$.  The theorem then follows from the
rank-nullity theorem.

A vector $v$ is in the nullspace of $\sqrt{p}I+A$ if and only if $Av=-\sqrt{p}v$, meaning that $v$ is an eigenvector of
$A$ with eigenvalue $-\sqrt{p}$.  Thus the nullity of $\sqrt{p}I+A$ is equal to the geometric multiplicity of $-\sqrt{p}$
as an eigenvalue of $A$.  The geometric multiplicity of $-\sqrt{p}$ is in turn at most the algebraic multiplicity of
$-\sqrt{p}$ as a root of the characteristic polynomial $q(x)=\det(xI-A)$ of $A$.  The characteristic polynomial
$q(x)$ has all coefficients in $\F$ as all entries of $A$ are in $\F$.  Moreover, $q(x)$ is a polynomial of degree at
most $N$ and so has at most $N$ roots.  Applying \thmref{thm:key}, we see that the multiplicity of $-\sqrt{p}$ can be at
most $\floor{\tfrac{N}{2}}$ as it occurs with the same multiplicity as $\sqrt{p}$.
\end{proof}

\subsection{Application to the correlation polytope}\label{candidate}
A great insight of \cite{FioriniMassarPokuttaTiwaryDewolf2012} is to identify a concrete hard submatrix of the
slack matrix of the correlation polytope.  The
submatrix of the slack matrix of $\text{COR}_n$ they consider is $B_n(x,y) = (x^Ty-1)^2$ for $x,y \in \{0,1\}^n$.
This matrix is an instance of \emph{unique disjointness}---an entry is $1$ when strings are disjoint,
and $0$ when strings have a unique intersection.  Results from communication complexity
\cite{Razborov92, deWolf00} show that this matrix has exponential nonnegative rank, giving the desired lower
bound on linear extended formulation size.

For PSD-rank, however, this matrix is not a suitable candidate as $A_n=[x^Ty-1]_{x,y \in \{0,1\}^n}$
satisfies $A_n \circ A_n = B_n$ and has rank at most $n+1$.

We will instead consider the $2^n$-by-$2^n$ matrix $M_n$ defined as $M_n(x,y)=(x^Ty-1)(x^Ty-2)$ for
$x,y \in \{0,1\}^n$, proposed as a candidate to have large PSD-rank by one of the authors \cite{Lee12}.
This matrix still enjoys the unique disjointness property, but is no longer obviously
the entrywise square of a low rank matrix.  We show that in fact the square root rank of $M_n$ is exponential.
First, let us verify that it is a submatrix of the slack
matrix of the correlation polytope.
\begin{lemma}\label{lem:slack} The $2^n$-by-$2^n$ matrix
$M_n=[(x^Ty-1)(x^Ty-2)]_{x,y \in \{0,1\}^n}$ is a submatrix of the slack matrix of the
correlation polytope $\text{COR}_n$.
\end{lemma}
\begin{proof}
For strings $x,y \in \{0,1\}^n$ note that $\Tr(xx^Tyy^T)=(x^Ty)^2$, and $\Tr(\diag(x)yy^T)=x^Ty$,
where $\diag(x)$ is the diagonal matrix whose diagonal is
$x$. The polynomial $(z-1)(z-2)=z^2-3z+2$ is nonnegative on integers $z$, thus for any $x \in \{0,1\}^n$,
 vertex $yy^T$ of the correlation
polytope satisfies the linear inequality
\begin{equation}
\label{eq:slack}
(x^Ty-1)(x^Ty-2)=\Tr((xx^T-3\diag(x))yy^T)+2\geq0.
\end{equation}
The entry $M_n(x,y)$ for $x, y \in \{0,1\}^n$ is thus the slack of the vertex $yy^T$ with the inequality
\eqnref{eq:slack} corresponding to $x$.
\end{proof}

For the lower bound on square root rank, we will actually work with a submatrix of $M_n$.  It is this matrix that
we will focus on for the remainder of the paper.
\begin{definition}
\label{def:prime_matrix}
Fix $n$ and let $p$ be the prime closest to $n/2$ (in case of a tie, pick the smaller one).  Define the matrix $P_n$ to be
the submatrix of $M_n$ restricted to strings of Hamming weight $p+1$.
\end{definition}
Note that the size of $P_n$ is $\binom{n}{p+1}$.  By Bertand's
Postulate (less well known as the Bertrand-Chebyshev theorem), for
any integer $m>1$, there is always at least one prime number $q$
such that $m<q<2m$. Choosing $m=\ceil{n/3-1}$, then there exists a
prime number in the interval $\left(\lceil n/3\rceil-1,2\cdot\lceil
n/3\rceil-2\right)$.  This shows that the size of $P_n$ is at least
$\binom{n}{\ceil{n/3}}$.

\begin{theorem}
\label{thm:main}
Let $n$ be a positive integer and let $N$ be the size of $P_n$.  Then
\[
\rootrank(P_n) \ge \ceil{\frac{N}{2}} \enspace .
\]
In particular, $\rootrank(P_n) \ge 3^{n/3-1}$.
\end{theorem}

\begin{proof}
Let $B$ be a matrix such that $B \circ B = P_n$.  We will lower bound the rank of $B$.

Note that all diagonal entries of $P_n$ are equal to
$p(p-1)$.  Thus all diagonal entries of $B$ are $\pm \sqrt{p(p-1)}$.
Further,  all off diagonal entries
of $P_n$ are of the form $s(s-1)$, where $s$ is an integer strictly smaller than $p$.

By multiplying $B$ on the left by a diagonal matrix $D$ whose diagonal entries are $\pm \tfrac{1}{\sqrt{p-1}}$ we
can obtain a matrix $C=DB$ with the same rank as $B$ and whose diagonal entries are all $\sqrt{p}$.  Further,
all off diagonal entries of $C$ are strictly less than $\sqrt{p}$.

Let $p_1, \ldots, p_t$ be an enumeration of all the primes strictly less than $p$, and let
$\F=\Q(\sqrt{p_1}, \ldots, \sqrt{p_t})$.  Note that $\sqrt{p} \not \in \F$ (see exercise 6.15 of \cite{Stewart04}).  On
the other hand, all off diagonal entries of $C$ are in $\F$.  Thus $C=\sqrt{p}I + A$ for a matrix $A$ with all entries
in $\F$.  Applying \thmref{thm:key} we find that the rank of $C$ is at least $\ceil{\tfrac{N}{2}}$.
\end{proof}

\section{An extension to more general decompositions}
We have now shown a lower bound on the square root rank of $P_n$.  In this section we see that this
lower bound can be leveraged into bounds on more general kinds of PSD factorizations.  We will look at
decompositions of the form
\begin{equation}
\label{eq:decomp}
M = \sum_{j=1}^{d^2} N_j \circ N_j \enspace.
\end{equation}
Let $k$ be the maximum rank of $N_i$ over $j \in [d^2]$ in such a decomposition.  If we can show that
$kd^2 > r$ for any decomposition as in~(\ref{eq:decomp}) then by \lemref{lem:decomposition} this would mean
the PSD-rank of $M$ is at least $r^{1/3}$.

We are able to do this provided certain restrictions are put on the matrices $N_i$.
Namely, we can show the following.
\begin{theorem}
\label{thm:extension}
Let $P_n$ be as in \defref{def:prime_matrix} and consider a decomposition of the form
\[
P_n=\sum_{j=1}^{d^2}(B_j\circ\sqrt{P_n})\circ(B_j\circ\sqrt{P_n}),
\]
where each matrix $B_j$ has rational entries.  Let $k$ the maximum
of $\rank(B_j \circ \sqrt{P_n})$ over $j \in [d^2]$.  Then
$kd^2 \ge \tfrac{1}{2} \binom{n}{\ceil{\frac{n}{3}}}$.
\end{theorem}

For the proof of the theorem we will use the following lemma.  We delay the proof of this lemma until after the
proof of the theorem.
\begin{lemma}\label{lem:sigma}
For any positive integer $\ell$, there are matrices with rational entries $\sigma_1, \ldots, \sigma_\ell$ each of
size $4^{\ceil{\ell/2}}$ such that for any real numbers $a_1, \ldots, a_k$
\[
\left(\sum_j a_j \sigma_j \right) \left( \sum_j a_j \sigma_j \right)
= \left(\sum_j a_j^2 \right) I_{4^{\ceil{\ell/2}}} \enspace .
\]
\end{lemma}

\begin{proof}[Proof of \thmref{thm:extension}]
Let $k$ be as in the theorem, and for simplicity assume that $d$ is even---the case where $d$ is odd
can be verified in the same way.  Let $N$ be the size of $P_n$.
Let $\sigma_1, \ldots,
\sigma_{d^2}$ be matrices defined in \lemref{lem:sigma} each of size $2^{d^2}$.

For each $j\in[d^2]$, we form a new
matrix $A_j = (B_j\circ\sqrt{P_n}) \otimes \sigma_j$. This matrix has
size $N2^{d^2}$. Further we let
\[
C= \sum_{j=1}^{d^2} A_j \enspace.
\]
Since each $B_j\circ\sqrt{P_n}$ is of rank at most $k$, it follows that the rank of $C$ will be
at most $kd^2\cdot2^{d^2}$.

We now lower bound the rank of $C$.  To do this we first define a
block diagonal matrix $D$ of size $N2^{d^2}$ with blocks of size $2^{d^2}$.
The $i^{th}$ diagonal block is defined as $\tfrac{1}{\sqrt{p-1}}\sum_j B_j(i,i) \sigma_j$.
By \lemref{lem:sigma}
\[
\left(\sum_j B_j(i,i) \sigma_j\right)\left(\sum_j B_j(i,i) \sigma_j\right)=I_{2^{d^2}}
\]
holds for every $i \in [N]$, thus the matrix
$D$ has full rank and $DC$ will have the same rank as $C$.  We will
actually lower bound the rank of $DC$.

We claim that the diagonal blocks of $DC$ are $\sqrt{p}\cdot I_{2^{d^2}}$.  Again by \lemref{lem:sigma},
the $i^{th}$ diagonal block of
$DC$ will be
\begin{align*}
\left (\frac{1}{\sqrt{p-1}} \sum_j B_j(i,i) \sigma_j\right)\left(\sum_j B_j(i,i)\sqrt{P_n(i,i)} \sigma_j\right)& =
\sum_j B_j(i,i)^2 \sqrt{\frac{P_n(i,i)}{p-1}} I_{2^{d^2}} \\
&=\sqrt{p} I_{2^{d^2}} \enspace .
\end{align*}

Now consider entries in the off diagonal blocks of $DC$.  As before let $p_1, \ldots, p_t$ be an
enumeration of the primes strictly less than $p$ and set $\F=\Q(\sqrt{p_1}, \ldots, \sqrt{p_t})$.  As the
$B_j$ and $\sigma_j$ matrices are rational, the off diagonal blocks of each $A_j$ have entries in
$\F$.  Further, $D$ is a matrix with entries in $\F$, thus the off diagonal blocks of $DC$ are also in $\F$.  As
$\sqrt{p} \not \in \F$ we can again apply \thmref{thm:rank} to conclude
$\rank(C) \ge \tfrac{1}{2} N 2^{d^2}$.  This implies $kd^2 \ge \tfrac{N}{2}$, which gives the theorem.
\end{proof}

\begin{proof}[Proof of \lemref{lem:sigma}]
Define
\[
X=\begin{bmatrix}
0&0&1&0 \\
0&0&0&1\\
1&0&0&0\\
0&1&0&0
\end{bmatrix},
Y=\begin{bmatrix}
0&0&0&1 \\
0&0&-1&0 \\
0&-1&0&0 \\
1&0&0&0
\end{bmatrix},
Z=\begin{bmatrix}
1&0&0&0\\
0&1&0&0\\
0&0&-1&0 \\
0&0&0&-1
\end{bmatrix} \enspace .
\]
These are real versions of the Pauli matrices.  They satisfy $XY=-YX, XZ=-ZX, YZ=-ZY$
and $X^2=Y^2=Z^2=I_4$.  Define
\begin{align*}
\sigma_{2j+1} &= Z^{\otimes j} \otimes Y \otimes I_4^{\otimes (\ceil{\ell/2} -j - 1} \\
\sigma_{2j} &= Z^{\otimes j} \otimes X \otimes I_4^{\otimes
(\ceil{\ell/2} -j - 1}
\end{align*}
Any $\sigma_i, \sigma_j$ for $i \ne j$ anti-commute, while
$\sigma_i^2=I_{4^{\ceil{\ell/2}}}$ which gives the property we
need.
\end{proof}

\section{Perspective}
There can be an unbounded gap between the square root rank and the
PSD-rank of a matrix.  Fawzi et al.
\cite{FawziGouveiaParriloRobinsonThomas14} gave an example of a
family of $k$-by-$k$ matrices with square root rank $k$ and PSD-rank
$2$.  Let $n_1, \ldots, n_k$ be an increasing sequence of integers
such that $2n_i-1$ is prime for every $i \in [k]$.  Define
$Q=[n_i+n_j-1]_{i,j \in [k]}$.  It can be easily seen that $Q$ has
normal rank and PSD-rank $2$, yet Fawzi et al. proved that the
square root rank is full.  This proof was the inspiration for our
\thmref{thm:main}.

We now give another example of a separation between square root rank and PSD-rank.  This example shows
the difficulties of generalizing our approach to show lower bounds on the PSD-rank itself.  Even
decompositions of the form studied in \thmref{thm:extension} have severe limitations.

Define a matrix indexed
by $x,y \in \{0,1\}^n$ as $F_n(x,y)=x^Ty(x^Ty-1)$.  This matrix is also a slack matrix of the correlation polytope
as can be verified by a very similar proof to \lemref{lem:slack}.
It can also be verified that the proof \thmref{thm:main} can be simply modified to show that $F_n$ has exponential
square root rank, and even
that the analogue of \thmref{thm:extension} holds for $F_n$.

On the other hand, the PSD-rank of $F_n$ is \emph{small}.  In fact, even the nonnegative rank of $F_n$ is small.
\begin{proposition}
\[
\rank_+(F_n) \le \binom{n}{2} \enspace .
\]
\end{proposition}

\begin{proof}
We recursively upper bound the rank of $F_n$.
The matrix $F_1$ is the all zero matrix and has nonnegative rank $0$.  Ordering the rows and columns of
$F_{n+1}$ by lexicographical order of $x \in \{0,1\}^n$ we can see
\[
F_{n+1} =
\begin{bmatrix}
F_n & F_n \\
F_n & F_n +D_n
\end{bmatrix}\enspace,
\]
where $D_n=[2x^Ty]_{x,y \in \{0,1\}^n}$. The matrix $D_n$ has nonnegative rank at most $n$.  Now using
$\rank_+(A+B) \le \rank_+(A) + \rank_+(B)$ and $\rank_+(A \otimes B) \le \rank_+(A) \rank_+(B)$ we find
$\rank_+(F_{n+1}) \le \rank_+(F_n)+n$.  Solving the recurrence gives $\rank_+(F_n) \le \binom{n}{2}$.
\end{proof}

This example shows that, while our bounds can be powerful for the square root rank, this approach is not likely
to give exponential lower bounds on the PSD-rank of the correlation polytope.  Indeed the techniques used in
\cite{LeeSR14} are quite different from those studied here.

\paragraph{Acknowledgments.}
Troy Lee and Zhaohui Wei are supported by the Singapore National
Research Foundation under NRF RF Award No.~NRF-NRFF2013-13.

\bibliographystyle{alpha}
\bibliography{nnegrk}

\end{document}